\documentclass{article}

\usepackage{graphicx}
\usepackage[utf8]{inputenc}
\usepackage[english]{babel}

\usepackage{amsthm}
\usepackage{amsmath}
\usepackage{amssymb}
\usepackage{amsfonts}
\usepackage{hyperref}

\def\R{\mathbb{R}}

\newtheorem{thm}{Theorem}
\newtheorem{lem}{Lemma}

\newtheorem{defin}{Definition}

\def\bajo#1#2{\smash{\mathop{#1}\limits_{#2}}}  

\title{Quantum codes do not fix qubit independent errors}

\author{J. Lacalle\thanks{Dep. de Matem\' atica Aplicada a las TIC, ETS de Ingenier\' ia de Sistemas Inform\' aticos, Universidad Polit\' ecnica de Madrid, C/ Alan Turing s/n, 28031, Madrid, Spain (jesus.glopezdelacalle\@upm.es).}, L.M. Pozo Coronado\thanks{Dep. de Matem\' atica Aplicada a las TIC, ETS de Ingenier\' ia de Sistemas Inform\' aticos, Universidad Polit\' ecnica de Madrid, C/ Alan Turing s/n, 28031, Madrid, Spain (lm.pozo@upm.es).}, A.L. Fonseca de Oliveira\thanks{Facultad de Ingenier\' ia, Universidad ORT, Montevideo, Uruguay (fonseca@ort.edu.uy).},\\ R. Mart\' in-Cuevas\thanks{Programa de Doctorado en Ciencias y Tecnolog\' ias de la Computaci\' on para Smart Cities, ETS de Ingenier\' ia de Sistemas Inform\' aticos, Universidad Polit\' ecnica de Madrid, C/ Alan Turing s/n, 28031, Madrid, Spain (r.martin-cuevas@alumnos.upm.es).}}

\date{}

\begin{document}

\maketitle

\begin{abstract}
In this work we prove that the $5-$qubit quantum error correcting code~\cite{BDSW,LMPZ} does not fix qubit independent errors (Theorem~\ref{Thm:QCodesDoNotFixErrors}), even assuming that the correction circuit does not introduce new errors. We say that a quantum code does not fix a quantum computing error if its application does not reduce the variance of the error. We also prove for qubit independent errors that if the correction circuit of the $5-$qubit quantum code detects an error, the corrected state has central symmetry (Theorem~\ref{Thm:Isotropy}) and, as a consequence, its variance is maximum (Lemma~\ref{Lem:VarianzaConError1}).

We have been able to obtain these results thanks to the high symmetry of the $5-$qubit quantum code and we believe that the necessary calculations for less symmetric codes are extremely complicated but that, despite this, the results obtained for the $5-$qubit quantum code reveal a general behavior pattern of quantum error correcting codes against qubit independent errors.

\end{abstract}

\section{Introduction}
It is well known that the main challenge to achieve an efficient quantum computation is the control of quantum errors~\cite{Ga}. To address this problem, two fundamental tools have been developed: quantum error correction codes~\cite{CS,St1,Go1,CRSS} in combination with fault tolerant quantum computing~\cite{Sh1,Pr,St3,Go2,KLZ,Ki,AB}.

In this article we study the effectiveness of the $5-$qubit quantum error correcting code~\cite{BDSW,LMPZ} to fix qubit independent quantum computing errors. We restrict the study to this specific quantum code because it is extremely difficult to perform the necessary calculations for more general quantum codes. Only the high degree of symmetry of the $5-$qubit quantum code allows carrying out the aforementioned calculations. However, we believe that the ability of this code to fix qubit independent errors will show a pattern of the behavior of general quantum codes against this type of quantum computing errors. Proof of this is the perfect adaptation of this code to qubit independent errors using the least number of qubits.

However, the analysis of qubit independent errors will be general. In order to do that, we represent $n-$qubits as points of the unit real sphere of dimension $d=2^{n+1}-1$~\cite{NC}, $S^d=\{x\in\R^{d+1}\ |\ \|x\| =1\}$, taking coordinates with respect to the computational basis $[|0\rangle,|1\rangle,\dots,|2^n-1\rangle]$,
\begin{eqnarray}
\label{QubitFormula}
\Psi=(x_0+ix_1,x_2+ix_3,\dots,x_{d-1}+ix_{d}).
\end{eqnarray}

Following previous works~\cite{LP,LPF}, we consider quantum computing errors as random variables with density function defined on $S^d$. As mentioned in these articles, it is easy to relate this representation to the usual representation in quantum computing by density matrices. In fact, if $X$ is a quantum computing error with density function $f(x)$, then the density matrix of $X$, $\rho(X)$, is obtained as follows, using the pure quantum states given by Formula (\ref{QubitFormula}):
$$
\rho(X)=\int_{S^d}f(x)|\Psi\rangle\langle\Psi|dx\quad\text{where}\quad\int_{S^d}f(x)dx=1.
$$

Density matrices do not always discriminate different quantum computing errors~\cite{NC}. Therefore, representations of quantum computing errors by random variables are more accurate than those by density matrices. Beside other considerations, while the space of random variables over $S^d$ is infinite-dimensional, the space of $n-$qubit density matrices has finite dimension. This is the main reason why the authors decided to use random variables to represent quantum computing errors. And once the representation of quantum computing errors is established by random variables, the most natural parameter to measure the size of quantum computing errors is the variance.

As described in~\cite{LP,LPF}, the variance of a random variable $X$ is defined as the mean of the quadratic deviation from the mean value $\mu$ of $X$, $V(X)=E[\|X-\mu\|^2]$. In our case, since the random variable $X$ represents a quantum computing error, the mean value of $X$ is the $n-$qubit $\Phi$ resulting from an errorless computation. Without loss of generality, we will assume that the mean value of every quantum computing error will always be $\Phi=|0\rangle$. To achieve this, it suffices to move $\Phi$ into $|0\rangle$ through a unitary transformation. Therefore, using the pure quantum states given by Formula (\ref{QubitFormula}), the variance of $X$ will be
\begin{eqnarray}
\label{VarCalculo}
V(X)=E[\|\Psi-\Phi\|^2]=E[2-2x_0]=2-2\int_{S^d}x_0f(x)dx.
\end{eqnarray}

In~\cite{LP} the variance of the sum of two independent errors on $S^d$ is presented for the first time. It is proved for isotropic errors and it is conjectured in general that
\begin{eqnarray}
\label{VarFormula}
V(X_1+X_2)=V(X_1)+V(X_2)-\frac{V(X_1)V(X_2)}{2}.
\end{eqnarray}

To relate the variance to the most common error measure in quantum computing, fidelity~\cite{NC}, the authors define a quantum variance that takes into account that quantum states are equivalent under multiplication by a phase. Thereby, the quantum variance of a random variable $X$ is defined as:
$$
V_q(X)=E[\bajo{\text{min}}{\phi}(\|\Psi-e^{i\phi}\Phi\|^2)]=2-2E\left[\sqrt{x_0^2+x_1^2}\right].
$$

The fidelity of the random variable $X$, $F(X)$, with respect to the pure quantum state $\Phi=|0\rangle$ satisfies $F(X)=\sqrt{\langle \Phi|\rho(X)|\Phi\rangle}$~\cite{NC}. Therefore, using the pure quantum states given by Formula (\ref{QubitFormula}), $F(X)^2=E[\langle \Phi|\Psi\rangle\langle \Psi|\Phi\rangle]=E[|\langle \Phi|\Psi\rangle|^2]=E[x_0^2+x_1^2]$. Now, the property $\sqrt{x_0^2+x_1^2}\geq x_0^2+x_1^2$ and Jensen's inequality $\sqrt{E[x_0^2+x_1^2]}\geq E[\sqrt{x_0^2+x_1^2}]$ allow us to conclude that:
$$
1-\dfrac{V_q(X)}{2}\leq F(X)\leq\sqrt{1-\dfrac{V_q(X)}{2}}.
$$

These inequalities show that quantum variance and fidelity are essentially equivalent, since when quantum variance tends to $0$, fidelity tends to $1$ and, conversely, when fidelity tends to $1$, quantum variance tends to $0$. Of the three measures, the variance is the only one that allows to complete the complicated calculations necessary to estimate the correction capacity of the 5-qubits code. It is also the standard measure in the statistical treatment of errors.

But the correct measure for quantum errors is the quantum variance, which as we have seen is equivalent to fidelity. However, we are going to see that the variance and the quantum variance have similar behaviors for the type of error that we want to analyze. Let $\Phi=|0\rangle$ be a qubit and suppose that $\Phi$ is changed by error becoming the state $\Psi=W\Phi$, where $W$ is the error operator given by the formula (\ref{For:W}) below whose density function $f(\theta_0)$ only depends on the angle $\theta_0$. Then:
$$
\begin{array}{lll}
\Psi & = & (\cos(\theta_0) + i\sin(\theta_0)\cos(\theta_1))\,|0\rangle + \\
     &   & (\sin(\theta_0)\sin(\theta_1)\cos(\theta_2) + i\sin(\theta_0)\sin(\theta_1)\sin(\theta_2))\,|1\rangle\ \text{and,} \\
\end{array}
$$
taking into account that
$$
\bajo{\text{min}}{\phi}(\|\Psi-e^{i\phi}\Phi\|^2)=2-2|\langle \Psi |\Phi \rangle|
$$
and the equation (\ref{VarCalculo}) we obtain:
$$
\begin{array}{lll}
V_q(X) & = & \displaystyle 2-4\pi \int_0^\pi \left(1 - \dfrac{\cos^2(\theta_0)}{2\sin(\theta_0)}\log\left(\dfrac{1-\sin(\theta_0)}{1+\sin(\theta_0)}\right) \right)\cdot f(\theta_0)\sin^2(\theta_0)\,d\theta_0 \\ \\
V(X)   & = & \displaystyle 2-4\pi \int_0^\pi 2\cos(\theta_0)\cdot f(\theta_0)\sin^2(\theta_0)\,d\theta_0 \\ \\
\end{array}
$$
We observe that the difference between the quantum variance and the variance are the weight functions of $f(\theta_0)\sin^2(\theta_0)$ in the integral and that they have a similar behavior for small errors, that is, for concentrated density functions $f(\theta_0)$ around $\theta_0=0$ (see Figure~\ref{Fig:WeightFunctions}).

\begin{figure}[th]
\label{Fig:WeightFunctions}\quad
\begin{center}
        \includegraphics[scale=0.25]{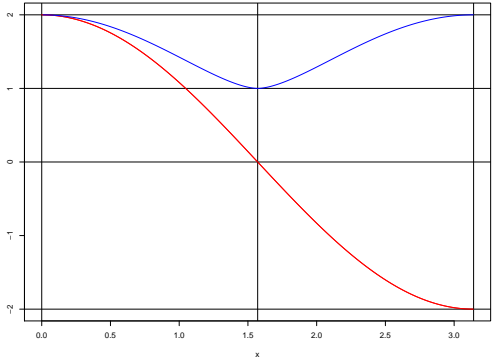}
        \caption{\centerline{Weight functions for quantum variance (red) and variance (blue).}}
\end{center}
\end{figure}

Even for large errors, for example a uniform distribution function $\displaystyle f=\dfrac{1}{2\pi^2}$, we have comparable values of the quantum variance and the variance:
$$
V_q (\Psi) = \dfrac{2}{3}\qquad \text{and}\qquad V(\Psi)=2.
$$

The study of the effectiveness of quantum error correcting codes has become essential to face the challenge of quantum computing. In this context, the problem we want to address is the following: Let $\Phi$ be an $5-$qubit encoded by the $5-$qubit quantum code $\mathcal C$. Suppose that the coded state $\Phi$ is changed by error, becoming the state $\Psi$. Now, to fix the error we apply the code correction circuit, obtaining the final state $\tilde\Phi$. While $\Phi$ is a pure state, $\Psi$ and $\tilde\Phi$ are random variables (mixed states). Our goal is to compare the variance of $\tilde\Phi$, $V(\tilde\Phi)=E[\|\tilde\Phi-\Phi\|^2]$, with that of $\Psi$, $V(\Psi)=E[\|\Psi-\Phi\|^2]$.

In order to compare the variances we will assume that the corrector circuit of $\mathcal C$ does not introduce new errors. In other words, we are going to estimate the theoretical capacity of the code to correct quantum computing errors. One would ideally expect that $\tilde\Phi=\Phi$ so that the variance of $\tilde\Phi$ would be $V(\tilde\Phi)=0$. Being more practical, we are only going to demand the minimum that could possibly be asked from an error correction process: $V(\tilde\Phi)<V(\Psi)$. If this minimum requirement is not met, we will say that the code $\mathcal C$ does not fix the corresponding quantum computing error.

The problem we address is, in our opinion, one of the biggest challenges for quantum computing. Consequently, it is also one of the most difficult tasks. For this reason, we restrict the problem in two ways. On the one hand, we consider the most widespread type of quantum error in the literature; qubit independent quantum computing errors. And on the other, we analyze the quantum code best adapted to this type of error; the $5-$qubit quantum code. This choice allows us to effectively compute the variances of the disturbed and corrected states, $\Psi$ and $\tilde\Phi$. The results obtained for this specific code reveal a general behavior pattern of quantum error correcting codes against qubit independent errors. The extension of the results of the $5-$qubit quantum code to general codes is analyzed in the conclusions of the article.

The results that we obtain are analogous to those presented in~\cite{LPF} for isotropic errors. But, although in both cases the results are as expected, what is surprising is their similarity despite the very different characteristics of the two types of error. The isotropic errors do not occur naturally and their density functions have a support of dimension $2^{n+1}-1$, while the qubit independent errors are commonly used to model quantum decoherence and their density functions have a much smaller dimension support, $4n$. Despite these great differences, these two types of errors present two analogies in relation to the ability of quantum codes to correct them: quantum error correcting codes do not fix these types of error and when the correction circuit of a quantum code detects one of these errors, the corrected state has the maximum variance and as a result, it already loses all the computing information.

The results are as expected because no quantum code can correct errors in all qubits simultaneously. The conclusions of this work and the one cited above~\cite{LPF} seem contradictory with the quantum threshold theorem, proved in the framework of fault-tolerant quantum computing~\cite{Sh1,Pr,St3,Go2,KLZ,Ki,AB}. The results are different because we use different error model. In fault-tolerant quantum computing the discretized quantum error model is used. In this model errors (that can be arbitrarily large) occur with a given probability $p$ and the probability that there is an error in $k$ qubits at the same time is $p^k$. Clearly this error model does not include the qubit independent continuous quantum errors that we use. Indeed in this model the probability of errors occurring in all qubits simultaneously is $1$ (that are small with high probability). We believe that it is necessary to develop the continuous quantum error model because the discrete one does not capture all the peculiarities of the real quantum computing errors.

The outline of the article is as follows: in Section 2 we set up the general structure of quantum error correcting codes; in Section 3 we analyze qubit independent quantum computing errors, we calculate the variance of the disturbed state $\Psi$ and we introduce the normal distribution over each qubit, as a particular case of type of error; in Section 4 we establish the result of applying the correction circuit of the $5-$qubit quantum code to a qubit independent error, we prove that this code does not fix qubit independent quantum computing errors and we analyze the behavior of a qubit independent error with normal distribution; finally, in Section 5 we analyze the results, we study if the behavior pattern of the $5-$qubit code can be extended to more general quantum codes and we comment on the problems that these results and those of article~\cite{LPF} pose to the viability of quantum computing.

\section{Quantum error correcting codes}

An quantum error correcting code of dimension $[n, m]$ is a subspace $\mathcal C$ of dimension $d^{\prime}=2^m$ in the $n-$qubit space ${\mathcal H}^n$, whose dimension is $d=2^n$. The $\mathcal C$ quantum code encoding function is a unitary operator $C$ that satisfies the following properties:
$$
C:\,{\mathcal H}^m\otimes{\mathcal H}^{n-m}\,\to\,{\mathcal H}^n \ \text{and}\ {\mathcal C}=C({\mathcal H}^m\otimes|0\rangle).
$$

The $\mathcal C$ code fixes $d^{\prime\prime}=2^{n-m}$ discrete errors: $E_0,\ E_1,\ \dots,\ E_{d''-1}$. Since the identity $I$ should be among these unitary operators, we assume that $E_0=I$. This process of discretization of errors allows to correct any of them if the subspaces $S_s=E_s({\mathcal C})$, $0\leq s < d^{\prime\prime}$, satisfy the following property:
\begin{equation}
\label{For:OrthogonalSum}
{\mathcal H}^n=S_0\,\bot\,S_1\,\cdots\,\bot\,S_{d''-1}.
\end{equation}
That is, ${\mathcal H}^n$ is the orthogonal direct sum of said subspaces. Note also that $S_0=E_0({\mathcal C})=I({\mathcal C})={\mathcal C}$. In the stabilized code formalism, the code $\mathcal C$ is the subspace of fixed states of an abelian subgroup of the Pauli group ${\mathcal P}_n=\{\pm 1,\pm i\}\times\{I,X,Z,Y\}^n$ and discrete errors are operators of ${\mathcal P}_n$ that anti-commute with any of the subgroup generators, except for the identity operator $E_0$. If Formula~(\ref{For:OrthogonalSum}) holds, the code is non-degenerate.

Suppose that a coded state $\Phi$ is changed by error, becoming the state $\Psi$. The initial state is a code state, that is, $\Phi\in S_0$, while the final state in general is not, that is, $\Psi\not\in S_0$. If the disturbed state belongs to the subspace $W_{\Phi}=L(E_0\Phi,\dots,E_{d''-1}\Phi)$, that is, if it is of the form
\begin{equation}
\label{For:FixableErrors}
\Psi=\alpha_0E_0\Phi+\cdots+\alpha_{d''-1}E_{d''-1}\Phi\quad\text{with}\quad|\alpha_0|^2+\cdots+|\alpha_{d''-1}|^2=1,
\end{equation}
then the quantum code allows us to retrieve the initial state $\Phi$. To achieve this, we measure $\Psi$ with respect to the orthogonal decomposition of the Formula~(\ref{For:OrthogonalSum}). The result will be $\frac{\alpha_s}{|\alpha_s|}E_s\Phi$ for a value $s$ between $0$ and $d^{\prime\prime}-1$. The value of $s$ is called syndrome and allows us to identify the discrete error that the quantum measurement indicates. Then, applying the quantum operator $E_s^{-1}$ we obtain $\frac{\alpha_s}{|\alpha_s|}\Phi$. This state is not exactly $\Phi$ but, differing only in a phase factor, both states are indistinguishable from the point of view of Quantum Mechanics. Therefore, the code has fixed the error.

An error that does not satisfy Formula~(\ref{For:FixableErrors}), that is, it does not belong to $W_{\Phi}$, cannot be fixed exactly. For example, if $\Psi$ belongs to the code subspace $\mathcal C$, the error cannot be fixed at all since, being a code state, it is assumed that it has not been disturbed. In this work we want to analyze the limitation in the correction capacity of an arbitrary code, assuming that the code correction circuit does not introduce new errors.

Finally, we want to highlight that discrete errors can be chosen so that, for example, all errors affecting a single qubit are fixed. The best code with this feature that encodes one qubit is the $5$-qubit quantum code~\cite{BDSW,LMPZ}. This code is optimal in the sense that no code with less than $5$ qubits can fix all the errors of one qubit and this article focuses on the ability of this code to fix qubit independent errors.

\section{Qubit independent quantum computing errors}

As in the previous section, let $\Phi$ be an initial code state, let $\mathcal C$ be the code with which $\Phi$ has been encoded and let $\Psi$ be the final state caused by an error on $\Phi$. This error can be modeled by means of a unitary operator $\mathcal W$. Therefore, the disturbed state will be $\Psi={\mathcal W}\Phi$. The model of qubit independent quantum errors is the simplest to represent decoherence in Quantum Computing. And, applied to a single qubit, it also allows us to model the error resulting from applying a quantum gate or a quantum measure to said qubit. Throughout this work we will use the following one-qubit error (unitary) operator:
$$
W=\left(\begin{array}{c}
\cos(\theta_0) + i\sin(\theta_0)\cos(\theta_1) \\
\sin(\theta_0)\sin(\theta_1)\cos(\theta_2) + i\sin(\theta_0)\sin(\theta_1)\sin(\theta_2) \\
\end{array}\right.
$$
\begin{equation}
\label{For:W}
\qquad\qquad\quad \left.\begin{array}{c}
-\sin(\theta_0)\sin(\theta_1)\cos(\theta_2) + i\sin(\theta_0)\sin(\theta_1)\sin(\theta_2) \\
\cos(\theta_0) - i\sin(\theta_0)\cos(\theta_1) \\
\end{array}\right),
\end{equation}
where $\theta_0,\theta_1\in[0,\pi]$ and $\theta_2\in[0,2\pi)$. The most important property of this operator is the following:
\begin{equation}
\label{For:PropertyW}
\|W\Phi-\Phi\|^2=2(1-\cos(\theta_0))\quad \text{for all qubit}\ \Phi.
\end{equation}
This fact implies that the variance of the error, which is defined as the expected value $E\left[\|W\Phi-\Phi\|^2\right]$ (see~\cite{LP,LPF}), is independent of the initial qubit $\Phi$. Formula~(\ref{For:PropertyW}) as well as many other results throughout the article is proved by mechanical calculations that are not made explicit. All of them are developed in~\cite{LPFM}.

Actually the error operator $W$ is a random variable that we are going to describe by means of its density function, whose natural domain is the three-dimensional sphere ${\mathcal S}_3$ parameterized in spherical coordinates by angles $\theta_0$, $\theta_1$ and $\theta_2$. From now on we will consider error operators $W$ whose density function depends exclusively on the first angle $\theta_0$. This fact greatly simplifies the analysis without losing generality, in view of the property indicated in Formula~(\ref{For:PropertyW}).

In the model of qubit independent quantum errors, the general $n-$qubit error operator is the following tensor product:
\begin{equation}
\label{For:ProductW}
{\mathcal W}=W_0\otimes\cdots\otimes W_{n-1},
\end{equation}
where $W_u$, $0\leq u<n$, is a one-qubit operator of the form described in Formula~(\ref{For:W}), with angles $\theta_0^u$, $\theta_1^u$ and $\theta_2^u$ and density function $f(\theta_0^u)$ defined on the sphere ${\mathcal S}_3$. Therefore ${\mathcal W}$ is a random variable defined on the space ${\mathcal S}_3^n = {\mathcal S}_3\times \cdots \times {\mathcal S}_3$.

Then, the disturbed state $\Psi={\mathcal W}\Phi$ is also a random variable. We are especially interested in its variance, as it was introduced in~\cite{LP,LPF}.

\begin{defin}
The variance of the random variable $\Psi$ is the expected value
$$
E\left[\|\Psi-\Phi\|^2\right].
$$
\end{defin}

Let $\tilde{\Phi}$ be the code state resulting from applying the code correction circuit to the state $\Psi$, assuming that this circuit does not introduce new errors. From the point of view of the statistical study of errors, the state $\tilde{\Phi}$ is a random variable. The random variable $\Psi$ describes the distribution of the error around $\Phi$. On the other hand, the random variable $\tilde{\Phi}$ describes the distribution of the error after correction around $\Phi$ in the code subspace $\mathcal C$ (since the accuracy of the correction circuit we are assuming implies that $\tilde{\Phi}$ belongs to $\mathcal C$).

In this context, in order to measure the code correction capacity, we compare the variance of the final error, variance of $\tilde{\Phi}$, with the variance of the initial error, variance of $\Psi$. Our study focuses on qubit independent quantum errors, and to describe the problem precisely, we study their variances.

\begin{thm}
The variance of a one-qubit quantum error with density function $f(\theta_0)$ is equal to
$$
V(\Psi) = 2 - 8\pi {\bar E}\left[\cos(\theta_0)\sin^2(\theta_0)\right],
$$
where $\displaystyle {\bar E}\left[\cos(\theta_0)\sin^2(\theta_0)\right]=\int_0^\pi f(\theta_0)\cos(\theta_0)\sin^2(\theta_0)d\theta_0$.
\label{Thm:VarianceOneQubit}
\end{thm}
\begin{proof}
In this case the error operator is ${\mathcal W}=W_0\otimes I\otimes\cdots\otimes I$ and, if we take an arbitrary initial state, then:
$$
\Phi=\sum_{k=0}^{2^n-1}\alpha_k|k\rangle\quad\text{and}\quad \Psi-\Phi=\sum_{k=0}^{2^n-1}\alpha_k\left((W-I)\otimes I^\prime\right)|k\rangle,
$$
where $I$ and $I^\prime$ are the identity operators for $1$ and $n-1$ qubits respectively and $W$ is the error operator given in Formula~(\ref{For:W}).

Splitting the index $k$ into the bit $k_0$ and the $n-1$ bit $k^\prime$ we obtain:
$$
\Psi-\Phi=\sum_{k^\prime=0}^{2^{n-1}-1}\left(\sum_{k_0=0}^{1}\alpha_{k_0,k^\prime}(W-I)|k_0\rangle\right)\otimes|k^\prime\rangle.
$$

The square of the norm of the above difference is:
$$
\|\Psi-\Phi\|^2=\sum_{k^\prime=0}^{2^{n-1}-1}\left\|\sum_{k_0=0}^{1}\alpha_{k_0,k^\prime}(W-I)|k_0\rangle\right\|^2.
$$

Finally, applying the property of Formula~(\ref{For:PropertyW}) we obtain:
$$
\|\Psi-\Phi\|^2=2\sum_{k^\prime=0}^{2^{n-1}-1} \left(|\alpha_{0,k^\prime}|^2 + |\alpha_{1,k^\prime}|^2\right) \left(1-\cos(\theta_0)\right)=2(1-\cos(\theta_0)).
$$

Now, the expected value is calculated by integrating into the sphere ${\mathcal S}_3$. This integral, in spherical coordinates, is:
\begin{eqnarray*}
V(\Psi) & = & 2 - 2 \displaystyle\int_{{\mathcal S}_3} f(\theta_0)\cos(\theta_0)d_{{\mathcal S}_3} \\
& = & 2 - 2 \displaystyle |{\mathcal S}_2| \int_0^\pi f(\theta_0)\cos(\theta_0)\sin^2(\theta_0) d\theta_0, \\
\end{eqnarray*}
where $\sin^2(\theta_0)$ is the volume element corresponding to $\theta_0$ in ${\mathcal S}_3$ and $|{\mathcal S}_2|$ is the volume of the sphere of dimension $2$.

Finally, the theorem is proved by substituting the value of $|{\mathcal S}_2|=4\pi$ (see Appendix) in the previous expression.
\end{proof}

\begin{thm}
\label{Thm:VarianceNQubit}
The variance of $\Psi={\mathcal W}\Phi$, where ${\mathcal W}$ is an $n-$qubit independent quantum error with density function $f(\theta_0^u)$ for all qubit $0\leq u<n$, is
$$
V(\Psi) = 2 - 2\left(4\pi {\bar E}\left[\cos(\theta_0)\sin^2(\theta_0)\right]\right)^n,
$$
where $\displaystyle {\bar E}\left[\cos(\theta_0)\sin^2(\theta_0)\right]=\int_0^\pi f(\theta_0)\cos(\theta_0)\sin^2(\theta_0)d\theta_0$.
\end{thm}
\begin{proof}
First of all we are going to prove that
\begin{equation}
\label{For:PropertyProductW}
\begin{array}{l}
\left\langle ({\mathcal W}-I)|j\rangle \big| ({\mathcal W}-I)|k\rangle \right\rangle = 2(1-\cos(\theta_0^0)\dots\cos(\theta_0^{n-1}))\delta_{j,k}+P \\ \\
\left\langle ({\mathcal W}-I)|j\rangle \big| k \right\rangle = (\cos(\theta_0^0)\dots\cos(\theta_0^{n-1})-1)\delta_{j,k} + P,
\end{array}
\end{equation}
for all $0\leq j,k<2^n$, where $P$ is, in each case, a polynomial whose variables are the sine and cosine functions of the angles $\theta_v^u$, $0\leq u<n$ and $0\leq v<3$, such that every monomial includes at least one of the following variables, with exponent one: $\cos(\theta_1^u)$, $\cos(\theta_2^u)$ or $\sin(\theta_2^u)$ with $0\leq u<n$. When calculating the expected value for the variance, the contribution of the polynomial $P$ will be zero.

We are going to prove the previous result by induction in the number $n$ of qubits. Base step: If $n = 1$ Formula~(\ref{For:PropertyW}) easily generalizes to:
\begin{equation}
\label{For:GeneralPropertyW}
\begin{array}{l}
\left\langle (W-I)|j\rangle \big| (W-I)|k\rangle \right\rangle=2(1-\cos(\theta_0^0))\delta_{j,k} \\ \\
\left\langle (W-I)|j\rangle \big| k \right\rangle=(\cos(\theta_0^0)-1)\delta_{j,k} + P,
\end{array}
\end{equation}
for all $0\leq j,k<2$ (See the proof in~\cite{LPFM}). In this case Property~(\ref{For:PropertyProductW}) holds.

Induction step: Let $n\geq 1$ and suppose by induction hypothesis that the Property~(\ref{For:PropertyProductW}) holds for $n$ qubits. We have to prove the Property~(\ref{For:PropertyProductW}) for $n+1$ qubits.

Splitting the $j$ and $k$ indices, as in the proof of Theorem~\ref{Thm:VarianceOneQubit}, we obtain:
$$
\begin{array}{l}
\left\langle ({\mathcal W}-I)|j\rangle \big| ({\mathcal W}-I)|k\rangle \right\rangle = \\ \\
\big\langle ({\mathcal W}^\prime \otimes W - I^\prime \otimes I)(|j^\prime\rangle \otimes |j_n\rangle) \ \big | \ ({\mathcal W}^\prime \otimes W - I^\prime \otimes I)(|k^\prime\rangle \otimes |k_n\rangle) \big\rangle = \\ \\
\big\langle ({\mathcal W}^\prime-I^\prime)|j^\prime\rangle \otimes (W-I)|j_n\rangle + |j^\prime\rangle \otimes (W-I)|j_n\rangle + ({\mathcal W}^\prime-I^\prime)|j^\prime\rangle \otimes |j_n\rangle \ \big| \\ \\
\ \, ({\mathcal W}^\prime-I^\prime)|k^\prime\rangle \otimes (W-I)|k_n\rangle + |k^\prime\rangle \otimes (W-I)|k_n\rangle + ({\mathcal W}^\prime-I^\prime)|k^\prime\rangle \otimes |k_n\rangle \big\rangle,
\end{array}
$$
where $j^\prime$ and $k^\prime$ are $n-$bits, $j_n$ and $k_n$ are bits, ${\mathcal W}^\prime$ is the error operator on $n$ qubits, $I^\prime$ is the $n-$qubit identity operator, $W$ is the error operator on one-qubit and $I$ the one-qubit identity operator.

Applying the distributive property in the previous expression, the sum of nine scalar products is obtained. Everyone decomposes as the product of two scalar products in which sets of disjoint angles appear. The induction hypothesis (Formula~(\ref{For:PropertyProductW})) and the base case (Formula~(\ref{For:GeneralPropertyW})) allow us to conclude the following:
$$
\begin{array}{l}
\left\langle ({\mathcal W}-I)|j\rangle \big| ({\mathcal W}-I)|k\rangle\right\rangle = \\ \\
\ \, (2(1-\cos(\theta_0^0)\dots\cos(\theta_0^{n-1}))\delta_{j^\prime,k^\prime}+P^\prime)2(1-\cos(\theta_0^n))\delta_{j_n,k_n} + \\ \\
\ \, ((\cos(\theta_0^0)\dots\cos(\theta_0^{n-1})-1)\delta_{j^\prime,k^\prime} + P^\prime)2(1-\cos(\theta_0^n))\delta_{j_n,k_n} + \\ \\
\ \, (2(1-\cos(\theta_0^0)\dots\cos(\theta_0^{n-1}))\delta_{j^\prime,k^\prime}+P^\prime)((\cos(\theta_0^n)-1)\delta_{j_n,k_n} + P) + \\ \\
\ \, ((\cos(\theta_0^0)\dots\cos(\theta_0^{n-1})-1)\delta_{k^\prime,j^\prime} + P^{\prime*})2(1-\cos(\theta_0^n))\delta_{j_n,k_n} + \\ \\
\ \, \delta_{j^\prime,k^\prime}2(1-\cos(\theta_0^n))\delta_{j_n,k_n} + \\ \\
\ \, ((\cos(\theta_0^0)\dots\cos(\theta_0^{n-1})-1)\delta_{k^\prime,j^\prime} + P^{\prime*})((\cos(\theta_0^n)-1)\delta_{j_n,k_n} + P) + \\ \\
\ \, 2((1-\cos(\theta_0^0)\dots\cos(\theta_0^{n-1}))\delta_{j^\prime,k^\prime}+P^\prime)((\cos(\theta_0^n)-1)\delta_{k_n,j_n} + P^*) + \\ \\
\ \, ((\cos(\theta_0^0)\dots\cos(\theta_0^{n-1})-1)\delta_{j^\prime,k^\prime} + P^\prime)((\cos(\theta_0^n)-1)\delta_{k_n,j_n} + P^*) + \\ \\
\ \, 2((1-\cos(\theta_0^0)\dots\cos(\theta_0^{n-1}))\delta_{j^\prime,k^\prime}+P^\prime)\delta_{j_n,k_n}. \\ \\
\end{array}
$$

Note that every product in which $P$, $P^\prime$ or both appears is an expression of type ``$P$'' and that the sum of expressions of this type is also of type ``$P$''. Furthermore, in every product in which neither $P$ nor $P^\prime$ is included as a factor, $\delta_{j^\prime,k^\prime}\delta_{j_n,k_n}=\delta_{j,k}$ appears. These facts allow us to obtain the following:
$$
\begin{array}{l}
\left\langle ({\mathcal W}-I)|j\rangle \big| ({\mathcal W}-I)|k\rangle\right\rangle = \\ \\
\qquad -2(1-\cos(\theta_0^0)\dots\cos(\theta_0^{n-1}))(1-\cos(\theta_0^n))\delta_{j,k} +  \\ \\
\qquad 2(1-\cos(\theta_0^n))\delta_{j,k} + 2(1-\cos(\theta_0^0)\dots\cos(\theta_0^{n-1}))\delta_{j,k} + P = \\ \\
\qquad 2(1-\cos(\theta_0^0)\dots\cos(\theta_0^{n-1})\cos(\theta_0^n))\delta_{j,k} + P.
\end{array}
$$

This result proves the induction step for the first of the Formula~(\ref{For:PropertyProductW}) equalities. The second equality is proved in an analogous way, from the expression:
$$
\begin{array}{l}
\left\langle ({\mathcal W}-I)|j\rangle \big| k \right\rangle = \big\langle ({\mathcal W}^\prime \otimes W -I^\prime \otimes I)(|j^\prime\rangle \otimes |j_n\rangle) \ \big|\ |k^\prime\rangle \otimes |k_n\rangle \big\rangle = \\ \\
\ \, \big\langle ({\mathcal W}^\prime -I^\prime)|j^\prime\rangle \otimes (W -I)|j_n\rangle + |j^\prime\rangle \otimes (W -I)|j_n\rangle + ({\mathcal W}^\prime -I^\prime)|j^\prime\rangle \otimes |j_n\rangle \ \big| \\ \\
\ \ \ \, |k^\prime\rangle \otimes |k_n\rangle \big\rangle. \\ \\
\end{array}
$$

This concludes the proof of the induction step and the equalities of Formula~(\ref{For:PropertyProductW}) are proved by the induction principle.

Now, the first equality of Formula~(\ref{For:PropertyProductW}) allow us to conclude the proof as follows:
$$
\begin{array}{lll}
V(\Psi) & = & \displaystyle E\left[\left\| \sum_{k=0}^{2^n-1} \alpha_k({\mathcal W}-I)|k\rangle \right\|^2\right] \\ \\
& = & \displaystyle E\left[ \sum_{j,k=0}^{2^n-1} \alpha_j^* \alpha_k \left\langle ({\mathcal W}-I)|j\rangle \big| ({\mathcal W}-I)|k\rangle \right\rangle \right] \\ \\
& = & \displaystyle \sum_{j,k=0}^{2^n-1} \alpha_j^* \alpha_k \left( E\left[ 2 - 2\cos(\theta_0^0)\dots\cos(\theta_0^{n-1}) \right]\delta_{j,k} + E\left[P\right]\right)  \\ \\
& = & \displaystyle \sum_{k=0}^{2^n-1} |\alpha_k|^2 E\left[ 2 - 2\cos(\theta_0^0)\dots\cos(\theta_0^{n-1}) \right] \\ \\
& = & 2 - 2E\left[\cos(\theta_0^0)\dots\cos(\theta_0^{n-1})\right] \\ \\
& = & \displaystyle 2 - 2\int_{{\mathcal S}_3}f(\theta_0^0)\cos(\theta_0^0)d_{{\mathcal S}_3}\cdots \int_{{\mathcal S}_3}f(\theta_0^{n-1})\cos(\theta_0^{n-1})d_{{\mathcal S}_3} \\ \\
& = & \displaystyle 2 - 2\left(|{\mathcal S}_2|\int_0^\pi f(\theta_0)\cos(\theta_0)\sin^2(\theta_0)d\theta_0\right)^n, \\ \\
\end{array}
$$
where $\sin^2(\theta_0)$ is the volume element corresponding to $\theta_0$ in ${\mathcal S}_3$ and $|{\mathcal S}_2|=4\pi$ is the volume of the sphere of dimension $2$, see Appendix.
\end{proof}

In Theorem~\ref{Thm:VarianceOneQubit} we have calculated the variance of the error over one qubit, and in Theorem~\ref{Thm:VarianceNQubit} the variance of the sum of the errors over each of the $n$ qubits. These errors are independent and identically distributed. In work~\cite{LP} the authors introduce the following formula for the sum of two independent errors $X_1$ and $X_2$:
$$
V(X_1+X_2)=V(X_1)+V(X_2)-\dfrac{V(X_1)V(X_2)}{2}.
$$
They prove it for isotropic errors and conjecture that it is generally true. They also generalize it for $n$ identically distributed independent errors as follows:
\begin{equation}
\label{For:VarianceSum}
V(X_1+\cdots +X_n)=2-2\left(1-\dfrac{\tau}{2}\right)^n,
\end{equation}
where $\tau$ is the variance of each of the errors. It is easy to check that in our case this formula is fulfilled. This result reinforces the conjecture about the variance of the sum of independent errors.

\subsection{An example}
\label{Sub:Example}

In this subsection we introduce a special distribution for the error of a qubit that has been key for~\cite{LP,LPF} and that illustrates well the results of this article.

\begin{defin}
\label{Def:Normal}
The normal error distribution for a qubit is one that has the following density function:
$$
f_n(\theta_0)=\frac{1}{2\pi^2}\frac{(1-\sigma^2)}{(1+\sigma^2-2\sigma\cos(\theta_0))^2},
$$
where the parameter $\sigma$ belongs to the interval $[0,1)$.
\end{defin}

When $\sigma$ approaches $1$ the probability is concentrated at the point $\Phi$, canceling the error. And when $\sigma$ approaches $0$ the distribution tends to be uniform, that is, a distribution in which after the disturbance all the states are equally probable. Figure~\ref{Fig:Distr} shows how the distribution changes depending on the parameter.

\begin{figure}[th]
\label{Fig:Distr}
\begin{center}
        \includegraphics[scale=0.25]{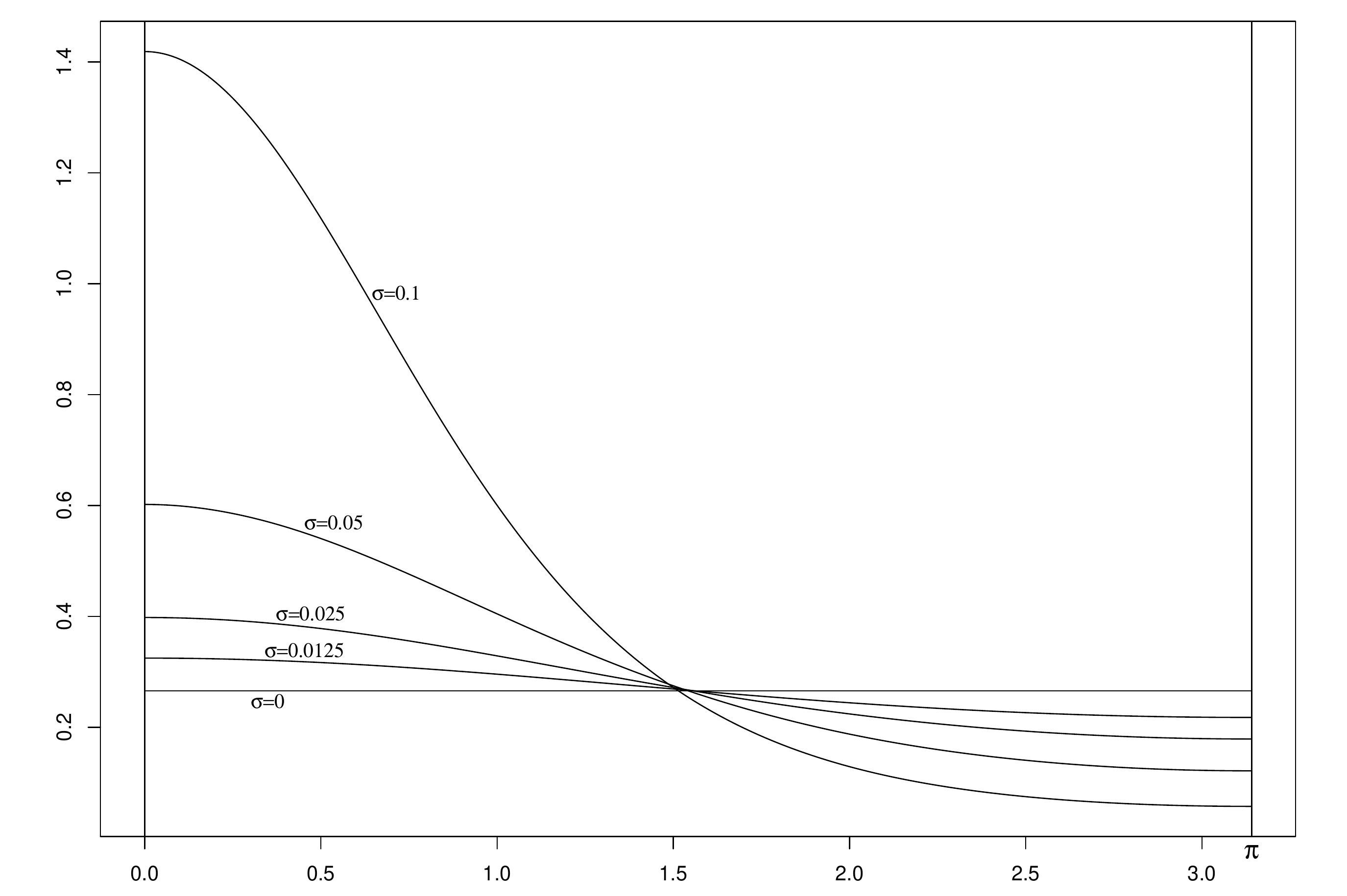}
        \caption{\centerline{Normal density function for a qubit.}}
\end{center}
\end{figure}

The variance of the normal distribution of a one-qubit is very simple: $V(\Psi)=2-2\sigma$ (see Theorem~\ref{Thm:VarianceOneQubit} and Appendix). And the variance of an $n-$qubit independent quantum errors with equal normal distributions in each qubit is, using Formula~(\ref{For:VarianceSum}), $V(\Psi)=2-2\sigma^n$.

\section{Result of applying the correction circuit of the $5-$qubit quantum code to a qubit independent error}

Let's call $\tilde{\Phi}$ the corrected state obtained by applying the correction circuit of code $\mathcal C$ to the disturbed state $\Psi$. Calculating the variance of the corrected state, $V(\tilde\Phi)$, for a general quantum code $\mathcal C$ is an extremely complicated task. Instead we will study the $5-$qubit quantum code~\cite{BDSW,LMPZ}. This code fixes all errors affecting a single qubit and, consequently, it will show key aspects of the power of quantum codes to control decoherence errors. On the other hand, the symmetry of this code allows to carry out the calculations.

The (unitary) error operator for independent qubit errors is, in this case,
$$
{\mathcal W}=W_0\otimes\cdots\otimes W_4,
$$
where the qubit errors, as described in Formula~(\ref{For:W}), are:
$$
W_u=\left(
\begin{array}{cc}
A_u+iB_u & -C_u+iD_u \\
C_u+iD_u & A_u-iB_u
\end{array}
\right),\qquad 0\leq u< 5,
$$
being
$$
\begin{array}{l}
A_u = \cos(\theta_0^u) \\
B_u = \sin(\theta_0^u)\cos(\theta_1^u) \\
C_u = \sin(\theta_0^u)\sin(\theta_1^u)\cos(\theta_2^u) \\
D_u = \sin(\theta_0^u)\sin(\theta_1^u)\sin(\theta_2^u) \\
\end{array},\qquad
\begin{array}{l}
\theta_0^u,\theta_1^u\in[0,\pi] \\
\theta_2^u\in[0,2\pi) \\
\end{array},
\qquad 0\leq u< 5.
$$
Since we are considering qubit independent errors, the density function of the error operator ${\mathcal W}$ will be defined on the space ${\mathcal S}_3^5$. Also, for simplicity we are considering equally distributed errors. Therefore the density function will be of the form
$$
f(\theta_0^0)f(\theta_0^1)f(\theta_0^2)f(\theta_0^3)f(\theta_0^4).
$$

The $5-$qubit quantum code $\mathcal C$ is defined by the following generators:
$$
\begin{array}{rcl}
|0_L\rangle & =   & |00000\rangle - |00011\rangle + |00101\rangle - |00110\rangle + \\
            &     & |01001\rangle + |01010\rangle - |01100\rangle - |01111\rangle - \\
            &     & |10001\rangle + |10010\rangle + |10100\rangle - |10111\rangle - \\
            &     & |11000\rangle - |11011\rangle - |11101\rangle - |11110\rangle \\ \\
|1_L\rangle & = - & |00001\rangle - |00010\rangle - |00100\rangle - |00111\rangle - \\
            &     & |01000\rangle + |01011\rangle + |01101\rangle - |01110\rangle - \\
            &     & |10000\rangle - |10011\rangle + |10101\rangle + |10110\rangle - \\
            &     & |11001\rangle + |11010\rangle - |11100\rangle + |11111\rangle , \\
\end{array}
$$
and by the discrete errors associated with the orthogonal direct sum given in Formula~(\ref{For:OrthogonalSum}):
$$
I,\,X_0,\,X_1,\,X_2,\,X_3,\,X_4,\,Y_0,\,Y_1,\,Y_2,\,Y_3,\,Y_4,\,Z_0,\,Z_1,\,Z_2,\,Z_3\ \text{and}\ Z_4,
$$
where $I$, $X$, $Y$ and $Z$ are the Pauli matrices
$$
I=\left(
\begin{array}{cc}
1 & 0 \\
0 & 1
\end{array}
\right),\quad
X=\left(
\begin{array}{cc}
0 & 1 \\
1 & 0
\end{array}
\right),\quad
Y=\left(
\begin{array}{cc}
0 & -i \\
i & 0
\end{array}
\right)\quad\text{and}\quad
Z=\left(
\begin{array}{cc}
1 & 0 \\
0 & -1
\end{array}
\right).
$$

A quantum state without error $\Phi$ is represented by $\Phi=(w_0+iw_1)|0_L\rangle+(w_2+iw_3)|1_L\rangle$, $w_0^2+w_1^2+w_2^2+w_3^2=1$, and the disturbed quantum state that describes the decoherence of $\Phi$ is the random variable $\Psi={\mathcal W}\Phi$.

To simplify the analysis of the action of the correction circuit of $\mathcal C$, we are going to change the computational basis $B_C$ in which the operator ${\mathcal W}$ is expressed by the basis associated with the code:
$$
B=[|0_L\rangle,\,|1_L\rangle,\,X_0|0_L\rangle,\,X_0|1_L\rangle,\,X_1|0_L\rangle,\,X_1|1_L\rangle,\dots,\,Z_4|0_L\rangle,\,Z_4|1_L\rangle].
$$
In this basis the expression of the error operator is
$$
{\mathcal W}_B = M^{\dag} {\mathcal W} M,
$$
where $M$ is the matrix that changes from basis $B$ to basis $B_C$. Its columns are the vectors of $B$ in coordinates of the basis $B_C$. The quantum state $\Phi$ at basis $B$ has the following coordinates:
$$
\Phi_B =(w_0+iw_1,\,w_2+iw_3,\,0,\,\cdots,\,0)_B.
$$
The disturbed quantum state $\Psi_B$ is obtained by multiplying $\Phi_B$ by the error operator: $\Psi_B={\mathcal W}_B\Phi_B$. Each of the entries in matrix ${\mathcal W}_B$ is a homogeneous polynomial of degree $5$ in the set of variables ${\mathcal V}=\{A_u,\,B_u,\,C_u,\,D_u\ |\ 0\leq u< 5\}$ and degree $1$ in the subsets of variables of ${\mathcal V}$ with subscript $u$ for all $0\leq u< 5$. Furthermore, each of these polynomials includes exactly $32$ monomials. This fact makes it necessary to use a symbolic calculation system to carry out the operations~\cite{LPFM}.

Analogously, the coordinates of $\Psi_B$ are homogeneous polynomials of degree $5$ in the set of variables ${\mathcal V}$, degree $1$ in the subsets of variables of ${\mathcal V}$ with subscript $u$ ($0\leq u< 5$) and degree $1$ in the variables $w_0,\ w_1,\ w_2$ and $w_3$. We are going to represent $\Psi_B$ as follows:
$$
\Psi_B = (\beta_0,\,\beta_1,\,\beta_2,\,\beta_3,\,\cdots,\,\beta_{30},\,\beta_{31})_B,
$$
and we are going to establish the properties of this quantum state (random variable) that will allow us to calculate the variance of the the corrected state $\tilde{\Phi}$. Note that we can consider the coordinates of $\Psi_B$ as functions of $w_0$, $w_1$, $w_2$ and $w_3$.

Note also that, given the definition of the basis $B$, the projection of $\Psi_B$ on the subspace $S_s$, $\Pi_s$, is:
\begin{equation}
\label{For:Proyector}
\Pi_s(\Psi_B)=\beta_{2s}|2s\rangle+\beta_{2s+1}|2s+1\rangle,\qquad 0\leq s< 16.
\end{equation}

\begin{lem}
\label{Lem:StatesOfSk}
The coordinates of the disturbed state $\Psi$ satisfy:
$$
\begin{array}{rcrcrl}
\beta_{2s}   & = & (a_s+ib_s)w_0  & + & (-b_s+ia_s)w_1 & + \\
              &   & (c_s+id_s)w_2  & + & (-d_s+ic_s)w_3 & \ \ \ , \qquad 0\leq s< 16, \\ \\
\beta_1      & = & (-c_0+id_0)w_0 & + & (-d_0-ic_0)w_1 & + \\
              &   & (a_0-ib_0)w_2  & + & (b_0+ia_0)w_3  & \ \ \ , \\ \\
\beta_{2s+1} & = & (c_s-id_s)w_0  & + &  (d_s+ic_s)w_1 & + \\
              &   & (-a_s+ib_s)w_2 & + & (-b_s-ia_s)w_3 & \ \ \ , \qquad 0< s< 16, \\
\end{array}
$$
where
$$
\begin{array}{ll}
a_s=\mathrm{Re}(\beta_{2s}(1,0,0,0)),\quad & b_s=\mathrm{Im}(\beta_{2s}(1,0,0,0)), \\
c_s=\mathrm{Re}(\beta_{2s}(0,0,1,0)),\quad & d_s=\mathrm{Im}(\beta_{2s}(0,0,1,0)),\qquad 0\leq s< 16.
\end{array}
$$
\end{lem}
\begin{proof}
See the proof in~\cite{LPFM}.
\end{proof}

Lemma~\ref{Lem:StatesOfSk} shows that the projection probabilities on the subspaces $S_s$, $0\leq s< 16$, do not depend on the initial state. This property is related to the property indicated in Formula~(\ref{For:PropertyW}) for the error in a qubit.

\begin{lem}
\label{Lem:ProbabilityOfPk}
The projections of the state $\Psi$ on the subspaces $S_s$, $0\leq s< 16$, satisfy:
$$
\|\Pi_s(\Psi)\|^2=a_s^2+b_s^2+c_s^2+d_s^2,
$$
where $a_s=\mathrm{Re}(\beta_{2s}(1,0,0,0))$, $b_s=\mathrm{Im}(\beta_{2s}(1,0,0,0))$, $c_s=\mathrm{Re}(\beta_{2s}(0,0,1,0))$ and $d_s=\mathrm{Im}(\beta_{2s}(0,0,1,0))$.
\end{lem}
\begin{proof}
It is enough to check, using Lemma~\ref{Lem:StatesOfSk}, the following equality (See the proof in~\cite{LPFM}):
$$
\|\Pi_s(\Psi)\|^2=|\beta_{2s}|^2+|\beta_{2s+1}|^2=(a_s^2+b_s^2+c_s^2+d_s^2)(w_0^2+w_1^2+w_2^2+w_3^2).
$$
The result is obtained taking into account that $w_0^2+w_1^2+w_2^2+w_3^2=1$.
\end{proof}

In order to carry out the study of the variance of $\tilde\Phi$, we need more information about the expressions $a_s$, $b_s$ $c_s$ and $d_s$, $0\leq s< 16$, introduced in Lemma~\ref{Lem:StatesOfSk} and used in the proof of Lemma~\ref{Lem:ProbabilityOfPk}.

\begin{lem}
\label{Lem:ShapeOfABCD}
The polinomials
$$
\begin{array}{ll}
a_s=\mathrm{Re}(\beta_{2s}(1,0,0,0)),\quad & b_s=\mathrm{Im}(\beta_{2s}(1,0,0,0)), \\
c_s=\mathrm{Re}(\beta_{2s}(0,0,1,0)),\quad & d_s=\mathrm{Im}(\beta_{2s}(0,0,1,0)),
\end{array}
$$
for all $0\leq s< 16$, satisfy the following properties:
\begin{itemize}
\item[1)] They are homogeneous polynomials of degree $5$ on the set of variables $\mathcal V$ and degree $1$ in the subsets of variables of ${\mathcal V}$ with subscript $u$, for all $0\leq u< 5$, and they are made up of $16$ monomials.
\item[2)] The monomials of $a_s$, $b_s$ $c_s$ and $d_s$ for all $0\leq s< 16$ are different two by two.
\item[3)] For all $0\leq s< 16$, in each of $a_s$, $b_s$ $c_s$ or $d_s$ every pair of different monomials share exactly one variable.
\end{itemize}
\end{lem}
\begin{proof}
See the proof in~\cite{LPFM}.
\end{proof}

When applying the correction circuit of code $\mathcal C$ to the disturbed state $\Psi$, the corrected state $\tilde{\Phi}$ is obtained in different ways, depending on the measured syndrome $s$, $0\leq s< 16$. We denote by $P_s$ the probability that the syndrome is $s$, $0\leq s< 16$. Thus $P_0$ is the probability that no error is detected, that is, that the quantum measurement projects the disturbed state $\Psi$ on the subspace $S_0={\mathcal C}$. On the other hand, $P_s$ is the probability that the $E_s$ error will be detected for $0< s< 16$.

The probability that the correction circuit of ${\mathcal C}$ measures a syndrome $s$ is, using Lemma~\ref{Lem:ProbabilityOfPk}:
$$
E[P_s]=E\left[\Pi_s(\Psi)\right]=E\left[a_s^2+b_s^2+c_s^2+d_s^2\right],\qquad 0\leq s< 16,
$$
where $a_s=\mathrm{Re}(\beta_{2s}(1,0,0,0))$, $b_s=\mathrm{Im}(\beta_{2s}(1,0,0,0))$, $c_s=\mathrm{Re}(\beta_{2s}(0,0,1,0))$ and $d_s=\mathrm{Im}(\beta_{2s}(0,0,1,0))$.

To calculate these expected values we need some preliminary considerations, that can be easily checked with a symbolic calculation system~\cite{LPFM}.

The first is related to the expected values of the combinations of two variables of ${\mathcal V}$:
\begin{equation}
\label{For:ValorEsperadoABCD}
\begin{array}{l}
\displaystyle E[A_u^2]=E[\cos^2(\theta_0)]=4\pi\bar E[\cos^2(\theta_0)\sin^2(\theta_0)] \\ \\
\displaystyle E[B_u^2]=E[\sin^2(\theta_0)\cos^2(\theta_1)]=\dfrac{4\pi}{3}\bar E[\sin^4(\theta_0)] \\ \\
\displaystyle E[C_u^2]=E[\sin^2(\theta_0)\sin^2(\theta_1)\cos^2(\theta_2)]=\dfrac{4\pi}{3}\bar E[\sin^4(\theta_0)] \\ \\
\displaystyle E[D_u^2]=E[\sin^2(\theta_0)\sin^2(\theta_1)\sin^2(\theta_2)]=\dfrac{4\pi}{3}\bar E[\sin^4(\theta_0)] \\ \\
\displaystyle E[A_uB_u]=E[A_uC_u]=E[A_uD_u]=0 \\ \\
\displaystyle E[B_uC_u]=E[B_uD_u]=E[C_uD_u]=0 \\ \\
\end{array}\ ,
\qquad 0\leq s< 5.
\end{equation}

The second has to do with the expected value of the polynomials $a_s^2$, $b_s^2$, $c_s^2$ and $d_s^2$, $0\leq s< 16$, that appear in Lemma~\ref{Lem:ProbabilityOfPk}.
Each of these expected values is equal to the sum of the expected values of the squared monomials of the corresponding polynomial. This is because the expected value of the product of different monomials is equal to zero. See Lemma~\ref{Lem:ShapeOfABCD}, item 3), and Formula~(\ref{For:ValorEsperadoABCD}), fifth and sixth lines.

The third is related to the number of monomials, classified by type, that appear in the polynomials $a_s^2$, $b_s^2$, $c_s^2$ and $d_s^2$ that contribute to the probability $P_s$, $0\leq s< 16$ (See the proof in~\cite{LPFM}):
\begin{equation}
\label{For:NumeroMonomiosPk}
\begin{array}{|c|c|c|c|c|c|c|}
\noalign{\hrule}
 & \ 5\ \text{$A$'s}\ & \ 4\ \text{$A$'s}\ & \ 3\ \text{$A$'s}\ & \ 2\ \text{$A$'s}\ & \ 1\ \text{$A$}\ & \ 0\ \text{$A$'s}\ \\
\noalign{\hrule}
P_0 & 1 & 0 & 0 & 30 & 15 & 18 \\
\noalign{\hrule}
\ P_s,\ 0< s< 16\ & 0 & 1 & 6 & 16 & 26 & 15 \\
\noalign{\hrule}
\end{array}
\end{equation}

The above considerations allow us to calculate the probabilities $E[P_s]$ for all $0\leq s< 16$.

\begin{thm}
\label{Thm:ProbabilidadPk}
The probability that the correction circuit of ${\mathcal C}$ measures a syndrome $s$, $0\leq s< 16$, is:
$$
\begin{array}{rccl}
E[P_0] & = & \displaystyle(4\pi)^5
                  & \displaystyle \left({\bar E}[\cos^2(\theta_0)\sin^2(\theta_0)]^5 + \dfrac{30}{3^3}{\bar E}[\cos^2(\theta_0)\sin^2(\theta_0)]^2{\bar E}[\sin^4(\theta_0)]^3 + \right. \\ \\
    &   &         & \ \ \displaystyle \left. \dfrac{15}{3^4}{\bar E}[\cos^2(\theta_0)\sin^2(\theta_0)]{\bar E}[\sin^4(\theta_0)]^4 + \dfrac{18}{3^5}{\bar E}[\sin^4(\theta_0)]^5\right) \\ \\
E[P_s] & = & \displaystyle(4\pi)^5
                  & \displaystyle \left( \dfrac{1}{3}{\bar E}[\cos^2(\theta_0)\sin^2(\theta_0)]^4{\bar E}[\sin^4(\theta_0)] + \right. \\ \\
    &   &         & \ \ \displaystyle \dfrac{6}{3^2}{\bar E}[\cos^2(\theta_0)\sin^2(\theta_0)]^3{\bar E}[\sin^4(\theta_0)]^2 + \\ \\
    &   &         & \ \ \displaystyle \dfrac{16}{3^3}{\bar E}[\cos^2(\theta_0)\sin^2(\theta_0)]^2{\bar E}[\sin^4(\theta_0)]^3 + \\ \\
    &   &         & \ \ \displaystyle \dfrac{26}{3^4}{\bar E}[\cos^2(\theta_0)\sin^2(\theta_0)]{\bar E}[\sin^4(\theta_0)]^4 + \\ \\
    &   &         & \ \ \displaystyle \left. \dfrac{15}{3^5}{\bar E}[\sin^4(\theta_0)]^5\right),\qquad 0< s< 16, \\
\end{array}
$$
where $\displaystyle {\bar E}\left[g(\theta_0)\right]=\int_0^\pi f(\theta_0)g(\theta_0)d\theta_0$ for every function $g$ defined in $[0,\pi]$.
\label{Thm:ProbabilityPs}
\end{thm}
\begin{proof}
It follows from the data in Table~(\ref{For:NumeroMonomiosPk}), applying the results of Formula~(\ref{For:ValorEsperadoABCD}).
\end{proof}

In order to calculate the variance of the corrected state ${\tilde \Phi}$ we introduce the corrected states ${\tilde \Phi}_s$, $0\leq s< 16$, that are obtained when the measured syndrome is $s$. Note that $\Psi$, ${\tilde \Phi}$ and ${\tilde \Phi}_s$ ($0\leq s< 16$) are random variables, all of them defined in the space $S_3^5$. Next we are going to prove a really surprising result. If the correction circuit for code ${\mathcal C}$ detects an error (measures a syndrome $0< s< 16$), then the probability distribution of $\tilde\Phi$ has central symmetry in subspace ${\mathcal C}$. This means that the probability of obtaining a $\tilde\Phi$ state is the same as the probability of obtaining its opposite $-\tilde\Phi$. This property implies that the variance of $\tilde\Phi$ will be the same as if the distribution of $\tilde\Phi$ were uniform (see~\cite{LPF}, Theorem 3). Therefore this result is a negative indicator of the correction power of the ${\mathcal C}$ code.

\begin{thm}
\label{Thm:Isotropy}
Random variables ${\tilde \Phi}_s$, $0< s< 16$, have central symmetric probability distributions in the code subspace ${\mathcal C}$.
\end{thm}
\begin{proof}
The probability distribution of
$
\displaystyle
{\tilde\Phi}_s=\dfrac{1}{\sqrt{P_s}}(\beta_{2s}|2s\rangle+\beta_{2s+1}|2s+1\rangle)
$
is
$
\displaystyle
\dfrac{P_s}{E[P_s]}
$.
We are going to prove that ${\tilde \Phi}_1$ has central symmetry. The proof of the remaining $14$ cases is analogous.

According to Lemma~\ref{Lem:StatesOfSk} the generators of $S_1$ are
$$
\begin{array}{l}
\beta_2=(a_1+ib_1)u_0 + (-b_1+ia_1)u_1 + (c_1+id_1)u_2  + (-d_1+ic_1)u_3\quad \text{and} \\
\beta_3=(c_1-id_1)u_0 + (d_1+ic_1)u_1  + (-a_1+ib_1)u_2 + (-b_1-ia_1)u_3\,,
\end{array}
$$
where
{\scriptsize
$$
\begin{array}{clcccccccc}
a_1 & =   & A_0A_1A_3B_2D_4 & + & A_0A_2A_4B_3D_1 & + & A_0B_1C_3C_4D_2 & + & A_0B_4C_1C_2D_3 & - \\
    &     & A_1A_2B_0B_4C_3 & + & A_1A_4C_0D_2D_3 & - & A_1B_3C_2C_4D_0 & + & A_2A_3C_0C_1C_4 & - \\
    &     & A_2B_1D_0D_3D_4 & - & A_3A_4B_0B_1C_2 & - & A_3B_4D_0D_1D_2 & - & A_4B_2C_1C_3D_0 & - \\
    &     & B_0B_2C_4D_1D_3 & - & B_0B_3C_1D_2D_4 & + & B_1B_2B_3B_4C_0 & + & C_0C_2C_3D_1D_4, & \\
b_1 & = - & A_0A_1B_3B_4D_2 & - & A_0A_2C_1C_3D_4 & - & A_0A_3C_2C_4D_1 & - & A_0A_4B_1B_2D_3 & + \\
    &     & A_1A_2A_3A_4D_0 & + & A_1B_0C_2D_3D_4 & - & A_1B_2C_0C_3C_4 & + & A_2B_0B_1B_3C_4 & - \\
    &     & A_2B_4C_0D_1D_3 & + & A_3B_0B_2B_4C_1 & - & A_3B_1C_0D_2D_4 & + & A_4B_0C_3D_1D_2 & - \\
    &     & A_4B_3C_0C_1C_2 & + & B_1B_4C_2C_3D_0 & + & B_2B_3D_0D_1D_4 & + & C_1C_4D_0D_2D_3, & \\
c_1 & =   & A_0A_1A_4C_2C_3 & + & A_0A_2A_3B_1B_4 & + & A_0B_2B_3C_1C_4 & + & A_0D_1D_2D_3D_4 & + \\
    &     & A_1A_2B_3C_0D_4 & - & A_1A_3B_0C_4D_2 & - & A_1B_2B_4D_0D_3 & - & A_2A_4B_0C_1D_3 & - \\
    &     & A_2C_3C_4D_0D_1 & + & A_3A_4B_2C_0D_1 & - & A_3C_1C_2D_0D_4 & - & A_4B_1B_3D_0D_2 & - \\
    &     & B_0B_1B_2C_3D_4 & - & B_0B_3B_4C_2D_1 & + & B_1C_0C_2C_4D_3 & + & B_4C_0C_1C_3D_2, & \\
d_1 & =   & A_0A_1A_2C_4D_3 & + & A_0A_3A_4C_1D_2 & + & A_0B_1B_3C_2D_4 & + & A_0B_2B_4C_3D_1 & + \\
    &     & A_1A_3B_4C_0C_2 & - & A_1A_4B_0B_2B_3 & - & A_1C_3D_0D_2D_4 & - & A_2A_3B_0D_1D_4 & + \\
    &     & A_2A_4B_1C_0C_3 & - & A_2B_3B_4C_1D_0 & - & A_3B_1B_2C_4D_0 & - & A_4C_2D_0D_1D_3 & - \\
    &     & B_0B_1B_4D_2D_3 & - & B_0C_1C_2C_3C_4 & + & B_2C_0C_1D_3D_4 & + & B_3C_0C_4D_1D_2. \\
\end{array}
$$
}

We are going to use changes in the angles that define the variables of the set $\mathcal V$ so that all monomials of $a_s$, $b_s$, $c_s$ and $d_s$, for a given $0\leq s< 16$, change sign. The changes that produce the sign flip of variables of $\mathcal V$ are:
\begin{equation}
\label{For:SignChanges}
\begin{array}{|c|c|}
\noalign{\hrule}
\text{Angle change} & \text{\ Variables whose sign changes\ } \\
\noalign{\hrule}
\theta_0^u \to \pi-\theta_0^u & A_u \\
\noalign{\hrule}
\theta_1^u \to \pi-\theta_1^u & B_u \\
\noalign{\hrule}
\ \theta_2^u \to (\pi-\theta_2^u)\ \text{mod}\ 2\pi\ & C_u \\
\noalign{\hrule}
\theta_2^u \to 2\pi-\theta_2^u & D_u \\
\noalign{\hrule}
\theta_2^u \to (\theta_2^u+\pi)\ \text{mod}\ 2\pi & C_u\ \text{and}\ D_u \\
\noalign{\hrule}
\end{array}\,,\quad 0\leq u< 5.
\end{equation}
Note that these changes allow to flip exactly the sign of any subset of variables of ${\mathcal V}$.

So, to change the signs of all the monomials of $a_s$, $b_s$, $c_s$ and $d_s$, $0\leq s< 16$, it is enough to find a subset of variables of $\mathcal V$, ${\mathcal V}_s$, such that each of the monomials includes an odd number of variables of ${\mathcal V}_s$. This subset exists for all $0\leq s< 16$ and if $s\neq 0$ the subset ${\mathcal V}_s$ can be chosen so that it does not include any variable $A_u$, for all $0\leq u< 5$.

For $s=1$ a subset ${\mathcal V}_1$ is for example the following (See all sets ${\mathcal V}_s$, $0<s<16$, and the proof that they verify the required property in~\cite{LPFM}):
$$
{\mathcal V}_1=\{\,B_3,\,B_4,\,C_0,\,C_2,\,D_0,\,D_2,\,D_3,\,D_4\,\}
$$
As we have said, changing the sign of the variables of ${\mathcal V}_1$ we change the sign of all monomials from $a_1$, $b_1$, $c_1$ and $d_1$. However, both $P_1$ and $E[P_1]$ remain invariant. In the first case, because each monomial of $P_1$ is the product of two monomials of $a_1$, $b_1$, $c_1$ or $d_1$ (see Lemma~\ref{Lem:ProbabilityOfPk}) and, therefore, they have an even number of variables that change sign. In the second case because $E[P_1]$ does not depend on the angles that have been modified (see Theorem~\ref{Thm:ProbabilidadPk}). This causes the resulting quantum state $\tilde\Phi_1$ to transform into $-\tilde\Phi_1$. Finally, since the density function only depends on the angles $\theta_0^u$, $0\leq u< 5$, the probabilities of $\tilde\Phi_1$ and $-\tilde\Phi_1$ are the same. This equality is maintained when integrating in space $S_3^5$ and, therefore, the resulting quantum state $\tilde\Phi_1$ has central symmetry.
\end{proof}

Finally we are able to calculate the variance of the state ${\tilde\Phi}_s$ resulting from the application of the correction circuit of the ${\mathcal C}$ code if the measured syndrome $s$ satisfies $0< s< 16$.

\begin{lem}
\label{Lem:VarianzaConError1}
The variance of the corrected state ${\tilde\Phi}_s$ ($0< s< 16$) is
$$
V({\tilde\Phi}_s)=2.
$$
\end{lem}
\begin{proof}
According to Theorem~\ref{Thm:Isotropy} the probability of obtaining the $\tilde\Phi_s$ state is the same as the probability of obtaining the $-\tilde\Phi_s$ state. This property allows us to calculate the variance of $\tilde\Phi_s$ as follows:
$$
\begin{array}{lcl}
V(\tilde\Phi_s) & = & \dfrac{1}{2}\left(V(\tilde\Phi_s)+V(-\tilde\Phi_s)\right) = \dfrac{1}{2}E\left[\|\tilde\Phi_s-\Phi\|^2+\|-\tilde\Phi_s-\Phi\|^2\right] \\ \\
                & = & \dfrac{1}{2}E\left[2\|\tilde\Phi_s\|^2+2\|\Phi\|^2\right] = \dfrac{1}{2}E\left[4\right] = 2. \\
\end{array}
$$
\end{proof}

In order to calculate $V(\tilde\Phi)$, it is only necessary to calculate the variance of the random variable ${\tilde\Phi}_0$, that is, the variance of the corrected state if the correction circuit of the ${\mathcal C}$ code does not detect any error.

\begin{lem}
\label{Lem:VarianzaConError0}
The variance of the state ${\tilde\Phi}_0$ is
$$
V({\tilde\Phi}_0)=2-2\dfrac{E[a_0]}{E[P_0]}+2\dfrac{E\left[a_0\left(1-\sqrt{P_0}\right)\right]}{E[P_0]},
$$
where $a_0=\mathrm{Re}(\beta_0(1, 0, 0, 0))$.
\end{lem}
\begin{proof}
The probability distribution of the state
$
\displaystyle
{\tilde\Phi}_0=\dfrac{1}{\sqrt{P_0}}(\beta_0|0\rangle+\beta_1|1\rangle),
$
is
$
\displaystyle
\dfrac{P_0}{E[P_0]}
$. Therefore its variance is:
$$
\begin{array}{lcl}
V({\tilde\Phi}_0) & = & E\left[\dfrac{P_0}{E[P_0]}\left\|\Phi-\dfrac{1}{\sqrt{P_0}}(\beta_0|0\rangle+\beta_1|1\rangle)\right\|^2\right] \\ \\
& = & \dfrac{1}{E[P_0]} E\left[\left\|\sqrt{P_0}\Phi-\beta_0|0\rangle-\beta_1|1\rangle\right\|^2\right] \\ \\
& = & \dfrac{1}{E[P_0]} \left(E\left[\left|\sqrt{P_0}(w_0+iw_1)-\beta_0\right|^2\right] + E\left[\left|\sqrt{P_0}(w_2+iw_3)-\beta_1\right|^2\right]\right). \\ \\
\end{array}
$$

We use Lemma~\ref{Lem:StatesOfSk} for the coordinates $\beta_0$ and $\beta_1$:
$$
\begin{array}{lcl}
V({\tilde\Phi}_0) & = & \dfrac{1}{E[P_0]} \left(E\left[\left(\sqrt{P_0}w_0-a_0w_0+b_0w_1-c_0w_2+d_0w_3\right)^2\right] + \right. \\ \\
&   & \qquad\quad\ \left. E\left[\left(\sqrt{P_0}w_1-b_0w_0-a_0w_1-d_0w_2-c_0w_3\right)^2\right] + \right. \\ \\
&   & \qquad\quad\ \left. E\left[\left(\sqrt{P_0}w_2+c_0w_0+d_0w_1-a_0w_2-b_0w_3\right)^2\right] + \right. \\ \\
&   & \qquad\quad\ \left. E\left[\left(\sqrt{P_0}w_3-d_0w_0+c_0w_1+b_0w_2-a_0w_3\right)^2\right]\right) \\ \\
& = & \dfrac{1}{E[P_0]} E\left[P_0+a_0^2+b_0^2+c_0^2+d_0^2-2a_0\sqrt{P_0}\right]\left(w_0^2+w_1^2+w_2^2+w_3^2\right) \\ \\
& = & \dfrac{1}{E[P_0]} E\left[2P_0-2a_0\sqrt{P_0}\right] = 2-2\dfrac{E\left[a_0\sqrt{P_0}\right]}{E[P_0]} \\ \\
& = & 2-2\dfrac{E[a_0]}{E[P_0]}+2\dfrac{E\left[a_0\left(1-\sqrt{P_0}\right)\right]}{E[P_0]}. \\
\end{array}
$$
We have used Lemma~\ref{Lem:ProbabilityOfPk} for the expression $a_0^2+b_0^2+c_0^2+d_0^2$.
\end{proof}

Finally we have everything necessary to calculate the variance of the corrected state $\tilde\Phi$.

\begin{thm}
\label{Thm:VarianzaEstadoCorregido}
The variance of the corrected state $\tilde\Phi$ is
$$
V(\tilde\Phi)=V(\Psi)+2E\left[a_0\left(1-\sqrt{P_0}\right)\right],
$$
where $a_0=\mathrm{Re}(\beta_0(1, 0, 0, 0))$.
\end{thm}
\begin{proof}
The variance of the corrected state satisfies:
$$
V(\tilde\Phi)=E[P_0]V({\tilde\Phi}_0)+\sum_{s=1}^{15}E[P_s]V({\tilde\Phi}_s).
$$

Applying Theorem~\ref{Thm:ProbabilidadPk} and Lemmas~\ref{Lem:VarianzaConError1} and~\ref{Lem:VarianzaConError0} the following is obtained:
$$
V(\tilde\Phi)=2(E[P_0]+15E[P_1])-2E[a_0]+2E\left[a_0\left(1-\sqrt{P_0}\right)\right].
$$

It is easily proven that
$$\begin{array}{lcl}
E[P_0]+15E[P_1] & = & (4\pi)^5\left(\bar E[\cos^2(\theta_0)\sin^2(\theta_0)] + \bar E[\sin^4(\theta_0)])^5\right)^5 \\ \\
& = & (4\pi)^5\bar E[\sin^2(\theta_0)]^5.
\end{array}
$$

The last expected value corresponds to the integral of the density function:
$$
\bar E[\sin^2(\theta_0)]=(4\pi)^{-1},\quad \text{therefore}\quad E[P_0]+15E[P_1]=1.
$$

The polynomial $a_0$ has $16$ monomials of which only $A_0A_1A_2A_3A_4$ has a non-zero expected value. The other $15$ are canceled because $E[B_u]=E[C_u]=E[D_u]=0$ for all $0\leq u< 5$. Then the expected value of $a_0$ is
$$
E[a_0]=(4\pi)^5\bar E[\cos(\theta_0)\sin^2(\theta_0)]^5.
$$

The proof is concluded by substituting $2-2(4\pi)^5\bar E[\cos(\theta_0)\sin^2(\theta_0)]^5$ for $V(\Psi)$, applying Theorem~\ref{Thm:VarianceNQubit}.
\end{proof}

\begin{thm}
\label{Thm:QCodesDoNotFixErrors}
The $5-$qubit quantum code does not fix qubit independent errors if the density function of the qubit error satisfies
\begin{equation}
\label{For:DensityFunctionCondition}
f(\theta_0)\geq f(\pi-\theta_0)\quad\text{for all}\quad \theta_0\in[0,\pi/2].
\end{equation}
\end{thm}
\begin{proof}
To conclude the proof it is enough to prove that $\displaystyle V(\tilde\Phi)-V(\Psi)\geq 0$. For this we must analyze the following variables:
{\scriptsize
$$
\begin{array}{clcccccccc}
a_0 & =   & A_0 A_1 A_2 A_3 A_4 & + & A_0 B_1 B_4 C_2 C_3 & + & A_0 B_2 B_3 D_1 D_4 & + & A_0 C_1 C_4 D_2 D_3 & + \\
    &     & A_1 B_0 B_2 C_3 C_4 & + & A_1 B_3 B_4 D_0 D_2 & + & A_1 C_0 C_2 D_3 D_4 & + & A_2 B_0 B_4 D_1 D_3 & + \\
    &     & A_2 B_1 B_3 C_0 C_4 & + & A_2 C_1 C_3 D_0 D_4 & + & A_3 B_0 B_1 D_2 D_4 & + & A_3 B_2 B_4 C_0 C_1 & + \\
    &     & A_3 C_2 C_4 D_0 D_1 & + & A_4 B_0 B_3 C_1 C_2 & + & A_4 B_1 B_2 D_0 D_3 & + & A_4 C_0 C_3 D_1 D_2, & \\
b_0 & =   & A_0 A_1 B_3 C_2 C_4 & + & A_0 A_2 B_1 D_3 D_4 & + & A_0 A_3 B_4 D_1 D_2 & + & A_0 A_4 B_2 C_1 C_3 & + \\
    &     & A_1 A_2 B_4 C_0 C_3 & + & A_1 A_3 B_2 D_0 D_4 & + & A_1 A_4 B_0 D_2 D_3 & + & A_2 A_3 B_0 C_1 C_4 & + \\
    &     & A_2 A_4 B_3 D_0 D_1 & + & A_3 A_4 B_1 C_0 C_2 & + & B_0 B_1 B_2 B_3 B_4 & + & B_0 C_2 C_3 D_1 D_4 & + \\
    &     & B_1 C_3 C_4 D_0 D_2 & + & B_2 C_0 C_4 D_1 D_3 & + & B_3 C_0 C_1 D_2 D_4 & + & B_4 C_1 C_2 D_0 D_3, & \\
c_0 & = - & A_0 A_1 C_3 D_2 D_4 & - & A_0 A_2 B_3 B_4 C_1 & - & A_0 A_3 B_1 B_2 C_4 & - & A_0 A_4 C_2 D_1 D_3 & - \\
    &     & A_1 A_2 C_4 D_0 D_3 & - & A_1 A_3 B_0 B_4 C_2 & - & A_1 A_4 B_2 B_3 C_0 & - & A_2 A_3 C_0 D_1 D_4 & - \\
    &     & A_2 A_4 B_0 B_1 C_3 & - & A_3 A_4 C_1 D_0 D_2 & - & B_0 B_2 C_1 D_3 D_4 & - & B_0 B_3 C_4 D_1 D_2 & - \\
    &     & B_1 B_3 C_2 D_0 D_4 & - & B_1 B_4 C_0 D_2 D_3 & - & B_2 B_4 C_3 D_0 D_1 & - & C_0 C_1 C_2 C_3 C_4, & \\
d_0 & =   & A_0 A_1 B_2 B_4 D_3 & + & A_0 A_2 C_3 C_4 D_1 & + & A_0 A_3 C_1 C_2 D_4 & + & A_0 A_4 B_1 B_3 D_2 & + \\
    &     & A_1 A_2 B_0 B_3 D_4 & + & A_1 A_3 C_0 C_4 D_2 & + & A_1 A_4 C_2 C_3 D_0 & + & A_2 A_3 B_1 B_4 D_0 & + \\
    &     & A_2 A_4 C_0 C_1 D_3 & + & A_3 A_4 B_0 B_2 D_1 & + & B_0 B_1 C_2 C_4 D_3 & + & B_0 B_4 C_1 C_3 D_2 & + \\
    &     & B_1 B_2 C_0 C_3 D_4 & + & B_2 B_3 C_1 C_4 D_0 & + & B_3 B_4 C_0 C_2 D_1 & + & D_0 D_1 D_2 D_3 D_4. & \\
\end{array}
$$
}
Note that by Lemma~\ref{Lem:ProbabilityOfPk} $P_0=a_0^2+b_0^2+c_0^2+d_0^2$.

We are going to prove that the expected value of each monomial of $a_0$ multiplied by $1-\sqrt{P_0}$ is greater than or equal to $0$.

For the first monomial of $a_0$, $A_0A_1A_2A_3A_4$, there are five sets of variables
$$
\begin{array}{rcl}
{\mathcal V}_0^\prime & = & \{\,A_0, B_0, B_3, B_4, C_2, D_2, D_3, D_4\,\}, \\
{\mathcal V}_1^\prime & = & \{\,A_1, B_1, B_2, B_3, C_4, D_2, D_3, D_4\,\}, \\
{\mathcal V}_2^\prime & = & \{\,A_2, B_0, C_0, C_3, C_4, D_2, D_3, D_4\,\}, \\
{\mathcal V}_3^\prime & = & \{\,A_3, B_1, B_2, B_4, C_2, D_1, D_3, D_4\,\}, \\
{\mathcal V}_4^\prime & = & \{\,A_4, B_1, C_1, C_2, C_3, D_2, D_3, D_4\,\} \\
\end{array}
$$
such that (See the proof that they verify the required property in~\cite{LPFM}):
\begin{itemize}
\item[a)] The intersection of each of the previous sets with the set of variables of any monomial of $a_0$, $b_0$, $c_0$ or $d_0$ has an odd cardinal.
\item[b)] The intersection of each of the previous sets with the set of variables of the monomial $A_0A_1A_2A_3A_4$ is $\{A_0\}$, $\{A_1\}$, $\{A_2\}$, $\{A_3\}$ and $\{A_4\}$ respectively.
\end{itemize}
We can associate to each one of these sets a change of variable, according to the pattern given in Formula~(\ref{For:SignChanges}), in such a way that only one of the variables of the monomial $A_0A_1A_2A_3A_4$ changes sign while the rest of the variables and $P_0$ remain unchanged.

This allows us to express the integral of $A_u(1-\sqrt{P_0})$ with respect to the angle $\theta_0^u$, $0\leq u< 5$, as follows:
$$
\begin{array}{rcl}
\displaystyle \int_0^\pi A_u(1-\sqrt{P_0})fd\theta_0^u & = & \displaystyle \int_0^{\pi/2} A_u(1-\sqrt{P_0})f(\theta_0^u)d\theta_0^u + \\ \\
 &   & \displaystyle \int_{\pi/2}^\pi A_u(1-\sqrt{P_0})f(\theta_0^u)d\theta_0^u \\ \\
 & = & \displaystyle \int_0^{\pi/2} A_u(1-\sqrt{P_0})f(\theta_0^u)d\theta_0^u - \\ \\
 &   & \displaystyle \int_0^{\pi/2} A_u(1-\sqrt{P_0})f(\pi-\theta_0^{\prime u})d\theta_0^{\prime u} \\ \\
 & = & \displaystyle \int_0^{\pi/2} A_u(1-\sqrt{P_0})(f(\theta_0^u)-f(\pi-\theta_0^u))d\theta_0^u. \\
\end{array}
$$
In the integral between $\pi/2$ and $\pi$ we have used the change of variable for $\theta_0^u$ given in Formula~(\ref{For:SignChanges}),
$$
\theta_0^u=\pi-\theta_0^{\prime u},
$$
and the properties mentioned above about the variable change associated to the set of variables ${\mathcal V}_u^\prime$.

Applying the previous result to all the variables of $A_0A_1A_2A_3A_4$ it is obtained that:
$$
E\left[A_0A_1A_2A_3A_4\left(1-\sqrt{P_0}\right)\right] = E[{\mathcal I}],\ \text{where}
$$
$$
\begin{array}{rl}
\displaystyle {\mathcal I}=\int_0^{\pi/2}\cdots \int_0^{\pi/2} A_0A_1A_2A_3A_4\left(1-\sqrt{P_0}\right) & (f(\theta_0^0)-f(\pi-\theta_0^0))\cdots \\ \\
 & (f(\theta_0^4)-f(\pi-\theta_0^4))d\theta_0^0\cdots d\theta_0^4. \\
\end{array}
$$
Given that in the integral ${\mathcal I}$ all the terms are greater than or equal to $0$, the expected value of ${\mathcal I}$ will be greater than or equal to zero and therefore it is satisfied
$$
E\left[A_0A_1A_2A_3A_4\left(1-\sqrt{P_0}\right)\right]\geq 0.
$$

The expected value of the remaining monomials of $a_0$ is equal to $0$. The proof in all cases is analogous, therefore we will only do it for the second monomial of $a_0$: $A_0B_1B_4C_2C_3$. For this monomial there is a set of variables
$$
{\mathcal V}_2^{\prime\prime} = \{\,A_1,B_1,B_2,B_3,C_4,D_2,D_3,D_4\,\}
$$
such that (See all sets ${\mathcal V}_s^{\prime\prime}$, $2\leq s<16$, and the proof that they verify the required property in~\cite{LPFM}):
\begin{itemize}
\item[a)] The intersection of $S$ with the set of variables of any monomial of $a_0$, $b_0$, $c_0$ or $d_0$ has an odd cardinal.
\item[b)] The intersection of $S$ with the set of variables of the monomial $A_0B_1B_4C_2C_3$ is $\{B_1\}$, that is, a set with a single variable that is not $A$.
\end{itemize}

This allows us to express the integral of $B_1(1-\sqrt{P_0})$ with respect to the angle $\theta_1^1$ as follows:
$$
\begin{array}{rcl}
\displaystyle \int_0^\pi B_1(1-\sqrt{P_0})fd\theta_1^1 & = & \displaystyle \int_0^{\pi/2} B_1(1-\sqrt{P_0})fd\theta_1^1 + \int_{\pi/2}^\pi B_1(1-\sqrt{P_0})fd\theta_1^1 \\ \\
 & = & \displaystyle \int_0^{\pi/2} B_1(1-\sqrt{P_0})fd\theta_1^1 - \int_0^{\pi/2} B_1(1-\sqrt{P_0})fd\theta_1^{\prime 1} \\ \\
 & = & 0. \\
\end{array}
$$
In the integral between $\pi/2$ and $\pi$ we have used the change of variable for $\theta_1^1$ given in Formula~(\ref{For:SignChanges}),
$$
\theta_1^1=\pi-\theta_1^{\prime 1},
$$
and the properties mentioned above about the variable change associated to the set of variables ${\mathcal V}_2^{\prime\prime}$.

Applying the previous result it is obtained that:
$$
E\left[A_0B_1B_4C_2C_3\left(1-\sqrt{P_0}\right)\right] = 0.
$$

Finally, it is proven that
$$
V(\tilde\Phi)-V(\Psi)=E\left[a_0\left(1-\sqrt{P_0}\right)\right]\geq 0.
$$
\end{proof}

Note that Formula~(\ref{For:DensityFunctionCondition}) is a sufficient condition for the $5-$qubit quantum code not to fix qubit independent errors. It is fulfilled by normal density functions $f:[0,\pi]\to\R$ such as, for example, non-increasing monotonic functions and functions with support $[0,\pi/2]$. However, Theorem~\ref{Thm:QCodesDoNotFixErrors} holds for more general density functions, regardless of whether they meet Formula~(\ref{For:DensityFunctionCondition}). It is enough that they satisfy that
$$
E\left[a_0\left(1-\sqrt{P_0}\right)\right]\geq 0.
$$

\subsection{Application to the example}

The normal distribution introduced in Definition~\ref{Def:Normal} satisfies the following properties:

$$
\begin{array}{l}
E[P_0] = \dfrac{1+15\sigma^8}{16}, \\ \\
E[P_s] = \dfrac{1-\sigma^8}{16} \quad\text{for all}\ 0< s< 16. \\ \\
\end{array}
$$

\noindent They are obtained from Theorem~\ref{Thm:ProbabilidadPk}, using the results of the Appendix.

On the other hand, this distribution function satisfies the equation (\ref{For:DensityFunctionCondition}) and consequently fulfills the theorem~\ref{Thm:QCodesDoNotFixErrors}, that is, the $5-$qubit quantum code does not fix qubit independent errors with this density function.

\section{Conclusions}

In this article we have analyzed the ability of the $5-$qubit quantum error correcting code to handle errors in quantum computing. We have presented a study similar to the one carried out by the authors in the case of isotropic quantum computing errors~\cite{LPF} and, despite the radically different characteristics of the two types of error, the results are surprisingly similar.

An important feature of quantum errors introduced in~\cite{LP}, in addition to variance, is the shape of their density functions, particularly the dimension of their supports (set of points in the domain where the density function is greater than zero). Thereby the support of a density function that represents $n-$qubit independent quantum computing errors has dimension $4n$, far from the support dimension of an isotropic error that is equal to $d=2^{n+1}-1$. However, the sum of quantum errors can increase the dimension of the support of the resulting error, since the dimension of the support of the sum of two independent errors in less than or equal to the sum of the dimensions of the respective supports. But this growth has a limit. The local quantum computing error model imposes a limitation on the growth of the dimension of the error support. For example, the support of $n-$qubit independent quantum computing errors has the maximum dimension indicated above.

Surprisingly the big difference in the support dimension of isotropic error versus qubit independent errors does not translate into a different behavior of the quantum error correcting codes against these two types of error. Both types of error have a global behavior that makes it difficult for quantum codes to control them. The final state $\tilde\Phi_s$ when an $s$ syndrome is detected, $s>0$, reaches the maximum variance $2$ in both cases, regardless of the error distribution function. In the case of isotropic error because the resulting distribution function is uniform (see~\cite{LPF}, Theorem 3) and in the case of qubit independent errors because the resulting distribution function is centrally symmetric (see Theorem~\ref{Thm:Isotropy} and Lemma~\ref{Lem:VarianzaConError1}). This means that if an error is detected in the code correcting circuit, in both cases the computing information has been completely lost. Furthermore, neither quantum codes fix isotropic errors (see~\cite{LPF}, Theorem 5) nor does the 5-qubit code fix independent qubit errors (see Theorem~\ref{Thm:QCodesDoNotFixErrors}).

In the analysis we have used the variance instead of the quantum variance due to the enormous difficulties that the latter poses in the calculations. The fact that the distribution of the corrected state is centrally symmetric if the code has detected an error, could lead us to think that when passing to the quantum variance, considering the phase factor, the error will be greatly reduced. However, this is not the case because the central symmetry property indicates that the probability distribution of the error in the corrected state is very close to being uniform. As we have analyzed in the introduction for error operators concentrated around the identity operator, the quantum variance and the variance will have similar behaviors.

In this work we have only analyzed the $5-$qubit quantum code. We have chosen it for its symmetry that has allowed us to complete the long and complicated calculations. But, to what extent the results obtained in this case are generally applicable? We believe that the pattern of behavior demonstrated for the $5-$qubit quantum code is general. A general quantum code, probably much less symmetric, will follow the pattern presented with small perturbations as occurs, for example, with the Fourier transform of a periodic function when passing from a domain that is a multiple of the period of the function to one that is not. Zero-width peaks at multiples of the frequency lose height and widen slightly due to loss of domain symmetry. In this sense, we are studying through numerical simulations the behavior of the Shor $9-$qubit quantum code~\cite{Sh2} and the $7-$qubit quantum code~\cite{St2}.

The results obtained in~\cite{LPF} and, especially, those presented in this article force us to rethink error control in quantum computing. The discretized model of errors that was so useful to design quantum codes seems not to be able to integrate all the subtleties of errors in quantum computing that are essentially continuous. Assuming that a discrete error occurs in a qubit with probability $p$, $0<p<1$, that a quantum code that corrects errors in a qubit leaves only errors with probability $p^2$, $p^3$... uncorrected and that the concatenation of codes allows us to further reduce the probability of uncorrected errors, is not enough to control subtle quantum errors. Actually the probability of errors occurring in all qubits at the same time is 1. Each of the errors is very small compared to a discrete error but they all occur simultaneously.

\section{Appendix}

The values of the integrals that have been used throughout the article are included in this appendix:

\vskip0.25cm \hskip1.5cm
$
\displaystyle \int_{0}^{\pi}\frac{\sin^2(\theta_0)}{(1+\sigma^2-2\sigma\cos(\theta_0))^d}d\theta_0=
\frac{\pi}{2}\frac{1}{(1-\sigma^2)},
$

\vskip0.25cm \hskip1.5cm
$
\displaystyle \int_{0}^{\pi}\frac{\cos(\theta_0)\sin^2(\theta_0)}{(1+\sigma^2-2\sigma\cos(\theta_0))^d}d\theta_0=
\frac{\pi}{2}\frac{\sigma}{(1-\sigma^2)},
$

\vskip0.25cm \hskip1.5cm
$
\displaystyle \int_{0}^{\pi}\frac{\sin^4(\theta_0)}{(1+\sigma^2-2\sigma\cos(\theta_0))^d}d\theta_0= \frac{3}{8}\pi,
$

\vskip0.25cm \hskip1.5cm
$
\displaystyle \int_{0}^{\pi}\frac{\cos^2(\theta_0)\sin^2(\theta_0)}{(1+\sigma^2-2\sigma\cos(\theta_0))^d}d\theta_0 = \frac{\pi}{8}\,\frac{1+3\sigma^2}{1-\sigma^2}.
$

\vskip0.25cm \noindent The surface of the unit spheres of dimensions 2 and 1 are $|{\mathcal S}_2| = 4\pi$ and $|{\mathcal S}_1| = 2\pi$.

\end{document}